\documentclass[conference,letterpaper]{IEEEtran}

\usepackage{tikz}
\usetikzlibrary{matrix}
\usepackage{graphicx}
\usepackage{epstopdf}
\usepackage{wrapfig}
\usepackage{amsfonts}
\usepackage{makecell}
\usepackage{pdfpages}
\usepackage[utf8]{inputenc} 
\usepackage{url}
\usepackage{ifthen}
\usepackage{cite}
\usepackage[cmex10]{amsmath}
\usepackage{mathtools}
\usepackage{dblfloatfix}
\usepackage{amsmath,amssymb,amsthm}
\usetikzlibrary{shapes,arrows}

\newtheorem{corollary}{Corollary}
\newtheorem{theorem}{Theorem}

\begin{document}

\title{Non-Asymptotic Converse Bounds Via Auxiliary Channels} 

\author{%
  \IEEEauthorblockN{Ioannis Papoutsidakis, Robert J. Piechocki, and Angela Doufexi}\\
  \IEEEauthorblockA{
  					Department of Electrical and Electronic Engineering,
                    University of Bristol\\  
                    Email: \{ioannis.papoutsidakis, r.j.piechocki, a.doufexi\}@bristol.ac.uk}
}

\maketitle

\begin{abstract}
This paper presents a new derivation method of converse bounds on the non-asymptotic achievable rate of discrete weakly symmetric memoryless channels. It is based on the finite blocklength statistics of the channel, where with the use of an auxiliary channel the converse bound is produced. This method is general and initially is presented for an arbitrary weakly symmetric channel. Afterwards, the main result is specialized for the $q$-ary erasure channel (QEC), binary symmetric channel (BSC), and QEC with stop feedback. Numerical evaluations show identical or comparable bounds to the state-of-the-art in the cases of QEC and BSC, and a tighter bound for the QEC with stop feedback.
\end{abstract}
\begin{IEEEkeywords}
Converse bounds, achievability bounds, finite blocklength regime, non-asymptotic analysis, channel capacity, discrete memoryless channels.
\end{IEEEkeywords}
\section{Introduction}
An important step towards latency mitigation is the study of achievable channel coding rates for finite blocklengths. Shannon proved in \cite{shannon1948mathematical} that there exists a code that can achieve channel capacity as the blocklength grows to infinity. This is the main reason the majority of conventional communications systems utilize blocks of several thousands of symbols to transmit with a rate approximately close to capacity. In \cite{poly}, several  lower bounds to capacity are established and a clearer image of the achievable coding rates in the finite blocklength regime is provided. Observing the available converse and achievability bounds it is apparent that the channel capacity significantly overestimates the achievable rates in the case of short blocklengths. 

The significance of finite blocklength results lies in the following: fewer channel uses prior to decoding allow not only for shorter transmission times but most importantly for less complex encoders and decoders as well as more flexible upper layer protocols. The trade-off between achievable rate and blocklength has been the focus of many researchers the past decade. Whether the achievability bounds of \cite{poly} are optimal remains an open question. This work aims to provide insight on the optimality of the state-of-the-art achievability and converse bounds, along with proposing a different approach of obtaining converse bounds over finite blocklengths. 

More specifically, the current paper formulates a novel method for the derivation of converse bounds on the non-asymptotic achievable rate of discrete weakly symmetric memoryless channels. This is accomplished by utilizing the finite statistics of the channel along with an auxiliary channel. The main result considers the average probability of error over arbitrary weakly symmetric memoryless channels and it is presented in Section \ref{genmeth}. Section \ref{prework} provides a short review on the state-of-the-art. In Sections \ref{qec}, \ref{bsc}, and \ref{vl}, the general method is specialized for the $q$-ary erasure channel (QEC), binary symmetric channel (BSC), and QEC with stop feedback, respectively. 

\textit{Notation:} Throughout this paper a $(n,M,\epsilon)$ code is a code with blocklength $n$, codebook size $M$, and average probability of error $\epsilon$. Furthermore, the letter $q$ denotes the cardinality of the channel input. Notation $X_1^n$ denotes the vector $X_1,...,X_n$. Finally, a $(l_a,M,n,\epsilon)$ VLSF code is a variable length stop feedback code with average blocklength $l_a$, codebook size $M$, packet size $n$, and average probability of error $\epsilon$.

\section{Previous Work}
\label{prework}

The best achievability bounds to date for discrete memoryless channels are established in \cite{poly}. Random Coding Union (RCU) bound is based on the analysis of error probability of random codes under maximum likelihood decoding \cite[Theorem 16]{poly}.
For an arbitrary $P_X$ there exists an $(M,\epsilon)$ code such that,
\begin{align}
\epsilon \leq \mathbb{E}[\min\{1, (M-1)\mathbb{P}[i(\bar{X};Y)\geq(X;Y)|X,Y] \}],
\end{align}
where $P_{XY\bar{X}}(a,b,c)=P_X(a)P_{Y|X}(b|a)P_X(c)$.

A less complex, in terms of computation, bound is the Dependence Testing (DT) bound \cite[Theorem 17]{poly}. For any distribution $P_X$ on $A$, there exists a code with $M$ codewords and average probability of error not exceeding
\begin{align}
\epsilon \leq \mathbb{E} \Bigg [\text{exp} \Bigg  \{-\bigg[i(X;Y)-\log  \frac{M-1}{2} \bigg]^+\Bigg \} \Bigg].
\end{align}
Numerical evaluations confirm that RCU bound is tighter for BSC and DT bound for BEC \cite{poly}.
%

Regarding the converse bounds, there are several older results from the early days of information theory. For instance, a bound based on Fano's inequality can be found in \cite{WOLFOWITZ19681}. The established sphere-packing bound as well as an improvement is given in \cite{SHANNON1967522} and \cite{4494709}, respectively. A special converse bound for the BEC with erasure probability $\delta$ is given in \cite[Theorem 38]{poly}, 
\begin{align}
\epsilon \geq \sum_{l=\lfloor{n-\log_2M}\rfloor+1}^n \binom{n}{l} \delta^l (1-\delta)^{n-l}\bigg(1-\frac{2^{n-l}}{M}\bigg),
\label{cnv3}
\end{align}
for a $(n,M,\epsilon)$ code. This bound corresponds to the meta-converse bound \cite{YP12}. Similarly, \cite[Theorem 35]{poly} provides a bound for the BSC that coincides with the meta-converse bound as well as the classical sphere-packing bound.
\section{Methodology}
\label{genmeth}
The stochasticity of a discrete memoryless channel is modelled by the conditional probability mass function of the channel output given the channel input $P_{Y|X}$ or vice versa $P_{X|Y}$. The nature of this process in the finite regime results to relative frequencies that do not always follow the probability mass function. For instance, a finite Bernoulli process with $P(X_i=1)=p$ has a positive probability of producing sequences where the relative frequency of $1$s is higher than $p$. Motivated by this observation a converse bound for weakly symmetric memoryless channels is developed.

The focus of this work is on weakly symmetric channels because the output of such channel can be defined  as the function $Y_1^n = f(X_1^n,Z_1^n)$, where $X_1^n$ is the channel input and $Z_1^n$ is a noise vector independent from $X_1^n$. Thus, by analysing the non-asymptotic behaviour of $Z_1^n$ an auxiliary channel can be produced in order to derive a converse to the non-asymptotic achievable rate of the main channel.

The main channel is depicted in Figure \ref{fstate} as a channel with random states where the conditional probability of the output given the input depends on the state. The state is not available at neither the encoder nor the decoder, therefore it is treated as part of the noise. The auxiliary channel is presented in Figure \ref{fstate2}. The noise is produced similarly to the main channel but in this case, there is side information at the receiver and the transmitter. That is to say, the state of the channel is known before transmission. 

Clearly, the auxiliary channel has greater achievable rate than the main one since there is no uncertainty about the state. Hence, a converse bound for the auxiliary channel bounds the achievable rate of the main channel as well. 

The question arises of how an auxiliary channel is produced based on the given main channel $P_{Y_1^n|X_1^n}$. A simple yet effective way to classify the channel states is based on the number of erroneous transmissions over one blocklength. In this manner, the states follow the binomial distribution,
\begin{align}
\begin{split}
P(S=s) = \binom{n}{s} &(1-P(Y=i|X=i))^s\\ &\cdot P(Y=i|X=i)^{n-s},
\end{split}
\label{bnm}
\end{align}
for any $i$ that belongs to the input alphabet.
 
At this stage, the capacity $C_s$ of each state of the auxiliary channel is calculated,
\begin{align}
C_s &= \max_{p(x_1,...,x_n)}{I(X_1^n;Y_1^n|S=s)}.
\end{align}
By setting a fixed aggregate rate $R_{a}=nR$ over all states, one can derive two different classes of states. One class includes the states with a capacity lower or equal to the rate $R_{a}$ and the other class the remaining. The latter class of states denotes cases where \textit{the error probability is certainly non-zero} since it contradicts the noisy channel coding theorem where $R < C$ for arbitrarily low probability of error. The probability of erroneous transmission due to rates that are not supported is a lower bound on the average probability of error given the rate $R_{a}$ for the auxiliary channel, and therefore for the main one as well. 

Formally, this is formulated as a corollary of the noisy channel coding theorem\cite{shannon1948mathematical}.
\begin{corollary}
For a weakly symmetric memoryless discrete channel, the average error probability of a $(n,M,\epsilon)$ code satisfies,
\begin{align*}
\epsilon \geq \sum_{s \in \mathcal{S}} P(\mathrm{error}|S=s)P(S = s),
\end{align*}
where,
\begin{align*}
\mathcal{S} =  \big \{s: C_s < \log_qM=nR \big \}.
\end{align*}
\label{corol}
\end{corollary}

To compute a lower bound on $P(\mathrm{error}|S=s)$ we can use two different approaches. The first and more general one is the strong converse to the noisy channel coding theorem by Wolfowitz \cite{578869}.

\begin{theorem}
For an arbitrary discrete memoryless channel of capacity C nats and any $(n,e^{nR},\epsilon)$ code with $R>C$,

\begin{align}
\epsilon \geq 1-\frac{4A}{n(R-C)^2}-e^{-\frac{n(R-C)}{2}},
\end{align}
where $A$ is a finite positive constant independent of $n$ or $R$.
\label{wolf}
\end{theorem}
As noted by the authors of \cite{poly}, Theorem \ref{wolf} is not useful for finite blocklength analysis, and indeed when used directly to the main channel it produces a converse bound which is looser than capacity. However, when it is combined with the method of auxiliary channels it can provide bounds that are comparable to the state-of-the-art as is demonstrated in the following sections.

The second approach is to take into consideration the specific structure of each channel state. Let $X_1^n=x_1^n$ be the decision of the optimal decoder given $Y_1^n=y_1^n$. Then,

\begin{align}
\begin{split}
P(\mathrm{error}|S=s) &= 1-P(X_1^n=x_1^n|Y_1^n=y_1^n) \\
&= 1-\frac{P(Y_1^n=y_1^n|X_1^n=x_1^n)P(X_1^n=x_1^n)}{P(Y_1^n=y_1^n)}\\
&=1-\frac{P(Z_1^n=z_1^n)P(X_1^n=x_1^n)}{P(Y_1^n=y_1^n)}\\
&=1-\frac{P(Z_1^n=z_1^n)P(X_1^n=x_1^n)}{\prod_{k=1}^n P(Y_k=y_k|Y_1^{k-1}=y_1^{k-1})}\\
&\geq 1-\frac{P(Z_1^n=z_1^n)P(X_1^n=x_1^n)}{\prod_{k=1}^n P(Y_k=y_k)}.
\end{split}
\label{secap}
\end{align}

Remarkably, the technique of using an auxiliary channel to derive results is not new in the field. Several network information theory problems are tackled in this manner. For instance, the characterisation of the sum-capacity of the Gaussian interference channel with weak interference, where a genie-aided channel is used to produce a converse \cite{gamal}. There are also some recent results where similar ideas are utilized for finite blocklength analysis \cite{9279290}. In the following sections, we produce tight and near-tight converse bounds for several channels using this method. 

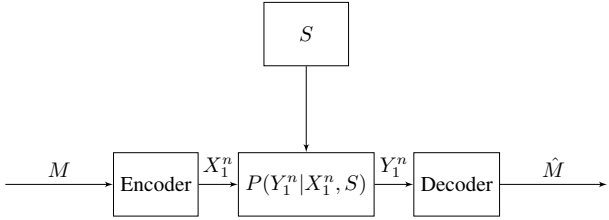
\begin{figure}
\scalebox{.8}{
\tikzstyle{block} = [draw, fill=white!20, rectangle, 
    minimum height=3em, minimum width=4em]
\tikzstyle{input} = [coordinate]
\tikzstyle{output} = [coordinate]
\tikzstyle{pinstyle} = [pin edge={to-,thin,black}]

\begin{tikzpicture}[auto, node distance=2.5cm,>=latex']
    \node [input, name=input] {};
    \node [block, right of=input] (Encoder) {Encoder};
    \node [block, right of=Encoder] (s) {$P(Y_1^n|X_1^n,S)$};
    \node [block, right of=s] (Decoder) {Decoder};
    \draw [->] (s) -- node[name=u] {$Y_1^n$} (Decoder);
    \node [output, right of=Decoder] (output) {};
    \node [block, above of=s] (state) {$S$};

    \draw [->] (input) -- node {$M$} (Encoder);
    \draw [->] (Encoder) -- node {$X_1^n$} (s);
    \draw [->] (Decoder) -- node [name=y] {$\hat{M}$}(output);
    \draw [->] (state) -- node {} (s);
\end{tikzpicture}}
\caption{The main channel presented as a channel with random states \cite{gamal}. Side informations about the state is not available at neither encoder nor decoder.}
\label{fstate}
\end{figure}

\begin{figure}
\scalebox{.8}{
\tikzstyle{block} = [draw, fill=white!20, rectangle, 
    minimum height=3em, minimum width=4em]
\tikzstyle{input} = [coordinate]
\tikzstyle{output} = [coordinate]
\tikzstyle{pinstyle} = [pin edge={to-,thin,black}]

\begin{tikzpicture}[auto, node distance=2.5cm,>=latex']
    \node [input, name=input] {};
    \node [block, right of=input] (Encoder) {Encoder};
    \node [block, right of=Encoder] (s) {$P(Y_1^n|X_1^n,S)$};
    \node [block, right of=s] (Decoder) {Decoder};
    \draw [->] (s) -- node[name=u] {$Y_1^n$} (Decoder);
    \node [output, right of=Decoder] (output) {};
    \node [block, above of=s] (state) {$S$};

    \draw [->] (input) -- node {$M$} (Encoder);
    \draw [->] (Encoder) -- node {$X_1^n$} (s);
    \draw [->] (Decoder) -- node [name=y] {$\hat{M}$}(output);
    \draw [dashed] (state) -| node[pos=0.99] {} 
        node [near end] {} (Encoder);
    \draw [dashed] (state) -| node[pos=0.99] {} 
        node [near end] {} (Decoder);
    \draw [->] (state) -- node {} (s);
\end{tikzpicture}}
\caption{Point-to-point communication system with state, where side information about the state is available at the encoder and the decoder. This setting plays the role of the auxiliary channel for the derivation of the converse bound.}
\label{fstate2}
\end{figure}
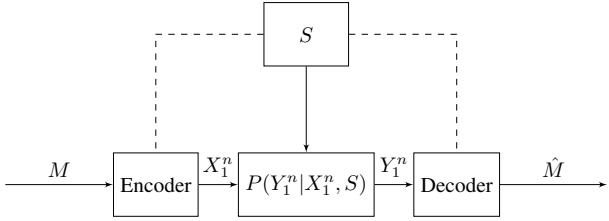

\section{$q$-ary Erasure Channel}
\label{qec}
Erasure channels are considered as a first illustration of the general method of Section \ref{genmeth}, since they are easy to manipulate in terms of mathematical analysis and intuition. When a $q$-ary symbol is transmitted over an erasure channel, the receiver obtains either the correct symbol with probability $1-\delta$ or an unknown symbol with probability $\delta$. 

%

Following the general method, the capacity of a channel with exactly $s$ erasures is computed as follows,
\begin{align}
C_s = n-s \text{ $q$-ary symbols/channel use.}
\label{ercap}
\end{align}
The proof can be found in Appendix \ref{appedA}.

As mentioned previously, the average probability of error $\epsilon$ is lower bounded by the probability of erroneous transmission due to unsupported rates. In the special case of QEC, the channel state is a linear function of the number of erasures. This allows for an easily established upper bound on the non-asymptotic achievable rate of the QEC, which can be composed as a lower bound to average probability of error. 

\begin{theorem}
For the $q$-ary erasure channel with erasure probability $\delta$, the average error probability of a $(n,q^{nR},\epsilon)$ code satisfies,
\begin{align}
\epsilon \geq \sum_{s=n-\lceil{nR}\rceil+1}^n \binom{n}{s} \delta^s (1-\delta)^{n-s}  \Big(1 - q^{n-s-nR}\Big).
\end{align}
\label{conver}
\end{theorem}

\begin{proof}
The channel states follow the binomial distribution as in (\ref{bnm}),
\begin{align*}
P(S=s)=\binom{n}{s} \delta^s (1-\delta)^{n-s}. 
\end{align*} 
For $nR > C_s$, the probability of error given the state is bounded based on (\ref{secap}) as follows,
\begin{align*}
\begin{split}
P(\mathrm{error}|S=s) &\geq 1-\frac{P(Z_1^n=z_1^n)P(X_1^n=x_1^n)}{\prod_{k=1}^n P(Y_k=y_k)}\\
&= 1 -\frac{P(Z_1^n=z_1^n)P(X_1^n=x_1^n)}{P(Z_1^n=z_1^n)\prod_{k \in \mathcal{K}} P(X_k=x_k)}\\
&= 1 - \frac{q^{-nR}}{q^{-(n-s)}}\\
&= 1 - q^{n -s-nR}.
\end{split},
\end{align*}
where $\mathcal{K}$ is the set of the indices of the unerased symbols.

For the bounds of the summation, it holds from Corollary \ref{corol} and (\ref{ercap}) that,
\begin{align*}
nR &> C_s = n-s \Rightarrow\\
nR &> n - s \Rightarrow\\
s &> n - nR \Rightarrow\\
s &> \lfloor{n - nR}\rfloor \Rightarrow\\
s &> n - \lceil{nR}\rceil \Rightarrow\\
 n - \lceil{nR}\rceil+1 \leq s &\leq n.
\end{align*} 
\end{proof}

\begin{theorem}
For the $q$-ary erasure channel with erasure probability $\delta$, the average error probability of a $(n,q^{nR},\epsilon)$ code satisfies,
\begin{align}
\begin{split}
\epsilon \geq &\sum_{s=n-\lceil{nR}\rceil+1}^n \binom{n}{s} \delta^s (1-\delta)^{n-s} \\& \cdot\Bigg(1-\frac{4A}{((nR-n+s)\ln(q))^2}-e^{-\frac{(nR-n+s)\ln(q)}{2}}\Bigg),
\end{split}
\end{align}
\label{conver2}
for any constant $A>0$.
\end{theorem}

\begin{proof}
The proof is similar to the proof of Theorem \ref{conver} with different derivation of the lower bound on $P(\mathrm{error}|S=s)$. For $nR > C_s$, the probability of error given the state  is bounded based on Theorem \ref{wolf} as follows,
\begin{align*}
\begin{split}
P(\mathrm{error}|S=s) &\geq 1-\frac{4A}{((nR-C_s)\ln(q))^2}-e^{-\frac{(nR-C_s)\ln(q)}{2}}\\
&=1-\frac{4A}{((nR-n+s)\ln(q))^2}\\
&\quad\quad-e^{-\frac{(nR-n+s)\ln(q)}{2}},
\end{split}
\end{align*}
for any constant $A>0$. Note that aggregate rate $nR$ and state capacity $C_s$ are converted to nats from $q$-ary symbols and the blocklength in Theorem \ref{wolf} is set to 1 since $n$ channel uses at the main channel are equivalent to one channel use for the auxiliary.

\end{proof}

Remarkably, Theorem \ref{conver} is the same as the meta-converse bound (\ref{cnv3}). It complies with the Singleton bound \cite{1053661} and consequently can be achieved by maximum distance separable (MDS) codes. Note also that there are not any non-trivial binary MDS codes that can achieve this bound for the binary erasure channel (BEC) \cite{1130706}. As a result, it is strictly greater than the achievable rate of BEC in non-trivial settings. Theorem \ref{conver2} produces a slightly relaxed bound as it is depicted in Figure \ref{f1}. This is expected since it does not consider the specific structure of each channel state.



\begin{figure}
\centering
\includegraphics[scale=0.6]{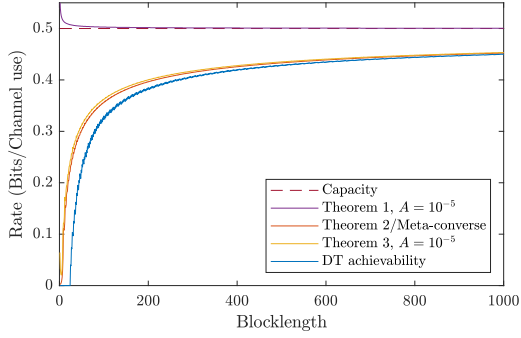}
\caption{Converse and achievability bounds on the non-asymptotic achievable rate of BEC(0.5) for average error probability $\epsilon = 10^{-3}$.}
\label{f1}
\end{figure}

\section{Binary Symmetric Channel}
\label{bsc}
Binary symmetric channel is another important binary-input channel where bits are inverted during transmission with probability $\delta$. Naturally, the differentiation of the states is based on the number of errors $s$. Hence, channel state capacities are measured as follows, 
\begin{align}
C_s = n-\log_2\binom{n}{s} \text{ bits/channel use.}
\label{ercap2}
\end{align}
The proof can be found in Appendix \ref{appedB}.

By manipulating (\ref{ercap2}) similarly to (\ref{ercap}), the lower bounds to average probability of error are established. Formally, 

\begin{theorem}
For the binary symmetric channel with error probability $\delta$, the average error probability of a $(n,2^{nR},\epsilon)$ code satisfies,
\begin{align}
\epsilon \geq \sum_{s \in \mathcal{S}} \binom{n}{s} \delta^s (1-\delta)^{n-s}\bigg(1 - \binom{n}{s}^{-1}2^{n-nR}\bigg),
\end{align}
where, 
\begin{align*}
\mathcal{S} =  \bigg \{s:\binom{n}{s}&>2^{n-nR}\wedge 0\leq s \leq n  \bigg \}.
\end{align*}
\label{convbsc}
\end{theorem}

\begin{proof}The channel states follow the binomial distribution as in (\ref{bnm}),\begin{align*} 
P(S=s)=\binom{n}{s} \delta^s (1-\delta)^{n-s}. 
\end{align*}
For $nR > C_s$, the probability of error given the state is bounded based on (\ref{secap}) as follows,
\begin{align*}
\begin{split}
P(\mathrm{error}|S=s) &\geq 1-\frac{P(Z_1^n=z_1^n)P(X_1^n=x_1^n)}{\prod_{k=1}^n P(Y_k=y_k)}\\
&= 1 - \frac{\binom{n}{s}^{-1}2^{-nR}}{2^{-n}}\\
&= 1 - \binom{n}{s}^{-1}2^{n-nR}.
\end{split}
\end{align*}

For the bounds of the summation, it holds from Corollary \ref{corol} and (\ref{ercap2}) that,
\begin{align*}
nR &> C_s =  n-\log_2\binom{n}{s} \Rightarrow\\
\log_2\binom{n}{s}&>n-nR\Rightarrow\\
\binom{n}{s}&>2^{n-nR}.\qedhere
\end{align*}
\end{proof}

\begin{theorem}
For the binary symmetric channel with error probability $\delta$, the average error probability of a $(n,2^{nR},\epsilon)$ code satisfies,
\begin{align}
\begin{split}
\epsilon \geq \sum_{s \in \mathcal{S}} \binom{n}{s} \delta^s (1-\delta)^{n-s}F(n,s,R,A),
\end{split}
\end{align}
\begin{align*}
F(n,s,R,A)=1&-\frac{4A}{((nR- n+\log_2\binom{n}{s} )\ln(2))^2}\\&\quad\quad-e^{-\frac{1}{2}(nR- n+\log_2\binom{n}{s} )\ln(2)},
\end{align*}
where, 
\begin{align*}
\mathcal{S} =  \bigg \{s:\binom{n}{s}&>2^{n-nR}\wedge 0\leq s \leq n  \bigg \},
\end{align*}
for any constant $A>0$.
\label{convbsc2}
\end{theorem}

\begin{proof}
The proof is similar to the proof of Theorem \ref{convbsc} with different derivation of the lower bound on $P(\mathrm{error}|S=s)$. For $nR > C_s$, the probability of error given the state  is bounded based on Theorem \ref{wolf} as follows,
\begin{align*}
\begin{split}
P(\mathrm{error}|S=s) &\geq 1-\frac{4A}{((nR-C_s)\ln(2))^2}-e^{-\frac{(nR-C_s)\ln(2)}{2}}\\
&=1-\frac{4A}{((nR- n+\log_2\binom{n}{s} )\ln(2))^2}\\&\quad\quad-e^{-\frac{1}{2}(nR- n+\log_2\binom{n}{s} )\ln(2)},
\end{split}
\end{align*}
for any constant $A>0$. Note that aggregate rate $nR$ and state capacity $C_s$ are converted to nats from bits and the blocklength in Theorem \ref{wolf} is set to 1 since $n$ channel uses at the main channel are equivalent to one channel use for the auxiliary.
\end{proof}

The numerical evaluation of Theorem \ref{convbsc} and Theorem \ref{convbsc2} is presented in Figure \ref{f1bsc} for the BSC with error rate $\delta = 0.05$ and average error probability $\epsilon = 10^{-3}$. Both converse bounds converge rapidly to the state-of-the-art meta-converse/sphere-packing bound.

\begin{figure}
\centering
\includegraphics[scale=0.6]{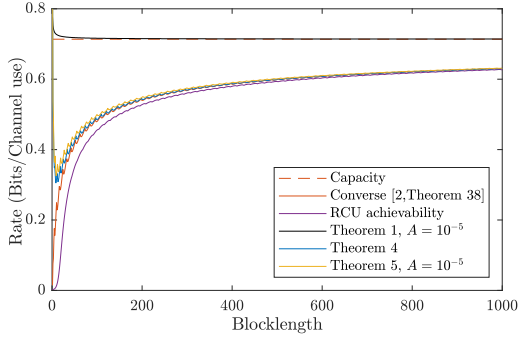}
\caption{Converse and achievability bounds on the non-asymptotic achievable rate of BSC(0.05) for average error probability $\epsilon = 10^{-3}$.}
\label{f1bsc}
\end{figure}

\section{$q$-ary Erasure Channel with Stop-Feedback}
\label{vl}
So far the method is demonstrated for channels with fixed blocklength where the resulting converse bounds are identical or similar to the state-of-the-art. In this section, a converse bound is produced for the QEC with stop feedback which is tighter than the state-of-the-art bound.

In the setting of variable length coding with stop feedback, rounds of packets of  $n$ symbols are transmitted. After the transmission of each packet the receiver decides whether it has enough channel outputs to perform decoding and informs the transmitter through the  feedback link on its decision \cite{5961844}.  The average number of required packets is denoted as $w_a$, hence the average blocklength is $l_a=w_an$.

The adaptation of the method of Section \ref{genmeth} to variable length coding with stop feedback utilizes an auxiliary channel as well. Contrary to the fixed blocklength setting, the aim is not to define a lower bound on average probability of error $\epsilon$ but a lower bound on average blocklength $l_a$ for $\epsilon=0$ given that feedback is noiseless and infinite transmissions prior to decoding are allowed. This means that the decoder is not allowed to randomly guess information bits when the rate is not supported by the channel state capacity, because it can result in a positive probability of error. By calculating the probability of the rate becoming supported by the channel at the $j$th channel use, the following bound is derived.

\begin{theorem}
The average blocklength of a $(l_a,q^k,n,0)$ VLSF code over a $q$-ary erasure channel with erasure rate $\delta$ is bounded as follows,

\begin{align}
l_a \geq n\sum_{m\in \mathbb{Z}^{+}} m\sum_{\substack{j={}\\mn-n+1}}^{mn} \binom{j-1}{\lceil{k}\rceil-1}\delta^{j-\lceil{k}\rceil}(1-\delta)^{\lceil{k}\rceil}.
\end{align}

\label{theoremvl}
\end{theorem}
\begin{proof}
Similarly to previous Sections, channel states are differentiated based on the number of erroneous transmission. Hence, channel state capacity (\ref{ercap}) stands. Since (\ref{ercap}) is always an integer, the criterion for the possibly supported rates can be the following,
\begin{align}
\begin{split}
R_a \leq C_s \Rightarrow \lceil{k}\rceil = j - s.
\end{split}
\end{align}
If the $\lceil{k}\rceil$th unerased symbol is received at the $j$th transmission, then the rate is supported and successful decoding might be possible. This process is described by the negative binomial distribution, 
\begin{align}
P(j;\delta,\lceil{k}\rceil) = \binom{j-1}{\lceil{k}\rceil-1}\delta^{j-\lceil{k}\rceil}(1-\delta)^{\lceil{k}\rceil}.
\end{align}
The probability of successful decoding during the $m$th packet is the following,
\begin{align}
\begin{split}
P(m;\delta,n,\lceil{k}\rceil) &= \sum_{j=mn-n+1}^{mn} P(j;\delta,\lceil{k}\rceil)\\
&= \sum_{\substack{j={}\\mn-n+1}}^{mn} \binom{j-1}{\lceil{k}\rceil-1}\delta^{j-\lceil{k}\rceil}(1-\delta)^{\lceil{k}\rceil}.
\end{split}
\end{align}
Hence, the average number of required packets for successful transmission is bounded as follows, 
\begin{align}
\begin{split}
w_a &\geq E[m] = \sum_{m\in \mathbb{Z}^{+}} mP(m;\delta,n,\lceil{k}\rceil)\Rightarrow\\
l_a &\geq  n\sum_{m\in \mathbb{Z}^{+}} mP(m;\delta,n,\lceil{k}\rceil)
\end{split}
\end{align}  
\end{proof}

In \cite{7606801}, several bounds for the BEC with feedback are presented. The authors discuss the gap between their achievability \cite[Theorem 7, Theorem 9]{7606801} and converse\cite[Corollary 6]{7606801} bounds which increases for very short average blocklengths. By comparing their results with Theorem \ref{theoremvl} in Figure \ref{fvl} it is becoming apparent that this gap is fundamental since Theorem \ref{theoremvl} coincides with the achievability bound \cite[Theorem 9]{7606801}. Additionally, Theorem \ref{theoremvl} is generally tighter than converse\cite[Corollary 6]{7606801}.

\begin{figure}
\centering
\includegraphics[scale=0.6]{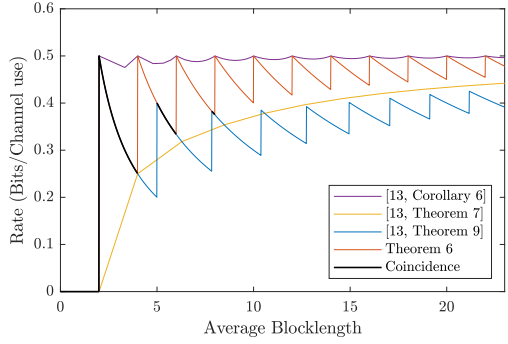}
\caption{Converse and achievability bounds on the non-asymptotic achievable rate of BEC(0.5) with noiseless stop feedback and average error probability $\epsilon = 0$.}
\label{fvl}
\end{figure}

\section{Conclusions}

The current paper presents novel results towards the characterization of the non-asymptotic achievable rate of memoryless discrete channels. A general derivation method of converse bounds is described for weakly symmetric channels and then it is particularized for the $q$-ary erasure channel, the binary symmetric channel, and the $q$-ary erasure channel with stop feedback. In the case of QEC, one of the derived bounds is identical to the meta-converse bound.  Numerical evaluations for BSC show quick convergence to other state-of-the-art converse bounds, namely the meta-converse and sphere-packing bounds. For the QEC with stop feedback, a bound that improves on the state-of-the-art is produced and it coincides with an achievability bound in certain settings. The extension of this result to arbitrary and asymmetric discrete memoryless channels is particularly interesting. Additionally, other settings with variable length coding with feedback can be explored. These coding schemes are very promising in the aspect of low-latency communications, however complex bounds are not easily specialized. Auxiliary channels with feedback could be a valuable tool for a better characterization of their achievable rates.

\section*{Acknowledgment}
This work is supported by the Engineering and Physical Sciences Research Council (EP/L016656/1); and the University of Bristol.


\appendices

\section{Proof of (\ref{ercap})}
\label{appedA}
The random variable $S$ denotes the number of erasures that occur after $n$ transmissions.
\begin{align}
\begin{split}
C_s &= \max_{p(x_1,...,x_n)}{I(X_1^n;Y_1^n|S=s)} \\
	&= \max_{p(x_1,...,x_n)}(H(X_1^n|S=s)-H(X_1^n|Y_1^n,S=s))\\
	&= \max_{p(x_1,...,x_n)}{\bigg(H(X_1^n)-s \frac{H(X_1^n)}{n}\bigg)}\\
	&= \max_{p(x_1,...,x_n)}{\bigg(H(X_1^n)\bigg(1-\frac{s}{n}\bigg) \bigg)}\\
	&= n-s \text{ $q$-ary symbols/channel use.}
\end{split}
\end{align}
The maximization of the state capacity is achieved by independent and identical uniform $q$-ary distributions on each channel input $X_i$. 

\section{Proof of (\ref{ercap2})}
\label{appedB}
The random variable $S$ denotes the number of errors that occur after $n$ transmissions.
\begin{align}
\begin{split}
C_s &= \max_{p(x_1,...,x_n)}{I(X_1^n;Y_1^n|S=s)} \\
	&= \max_{p(x_1,...,x_n)}(H(Y_1^n|S=s)-H(Y_1^n|X_1^n,S=s))\\
	&= \max_{p(x_1,...,x_n)}(H(X_1^n+Z_1^n)-H(X_1^n+Z_1^n|X_1^n))\\
	&= \max_{p(x_1,...,x_n)}(H(X_1^n+Z_1^n)-H(Z_1^n))\\
	&= n-\log_2\binom{n}{s} \text{ bits/channel use.}
\end{split}
\end{align}
Vector $Z_1^n$ denotes a binary random vector of Hamming weight $s$. Generally,  if $X_i$ follows a uniform distribution then $X_i+Z_i$ also follows a uniform distribution and its entropy is maximised.
\bibliography{bound_article}

\bibliographystyle{ieeetr}

\end{document}